\documentclass[12pt]{iopart}

\usepackage{mathrsfs}

\usepackage{graphicx}
\usepackage{dcolumn}
\usepackage{bm}

\expandafter\let\csname equation*\endcsname\undefined
\expandafter\let\csname endequation*\endcsname\undefined
\usepackage{amsmath}
\usepackage{mathtools}
\usepackage{amssymb}
\usepackage{amsthm}
\usepackage{xspace}
\usepackage{xfrac}
\usepackage[mathcal]{euscript}
\usepackage{tikz}
\usepackage{tikz-3dplot} 

\usepackage{orcidlink}

\usepackage{url}
\usepackage{hyperref}
\usepackage{color}
\definecolor{refcolor}{RGB}{0,0,190}
\hypersetup{
    colorlinks,
    citecolor=refcolor,
    filecolor=refcolor,
    linkcolor=refcolor,
    urlcolor=refcolor
}

\usepackage{booktabs}

\usepackage{bookmark}
\bookmarksetup{
  numbered, 
  open,
}

\theoremstyle{definition}
\newtheorem{theorem}{Theorem}

\newtheorem{lemma}{Lemma}

\theoremstyle{remark}
\newtheorem{remark}{Remark}

\theoremstyle{definition}

\newtheorem{problem}{Problem}
\newtheorem{definition}{Definition}

\newtheorem{objection}{Objection}

\newtheorem{question}{Question}

\theoremstyle{definition}
\newtheorem{defCustom}{Definition}
\renewcommand{\thedefCustom}{\arabic{definition}}
\makeatletter
\newcommand{\setdefCustomtag}[1]{
  \let\oldthedefCustom\thedefCustom
  \renewcommand{\thedefCustom}{#1}
  \g@addto@macro\enddefCustom{
    \global\let\thedefCustom\oldthedefCustom}
  }
\makeatother

\newtheorem{condition}{Condition}
\renewcommand{\thecondition}{\arabic{condition}}
\makeatletter
\newcommand{\setconditiontag}[1]{
  \let\oldthecondition\thecondition
  \renewcommand{\thecondition}{#1}
  \g@addto@macro\endcondition{
    \global\let\thecondition\oldthecondition}
  }
\makeatother

\newtheorem{assumption}{Assumption}
\renewcommand{\theassumption}{\arabic{assumption}}
\makeatletter
\newcommand{\setassumptiontag}[1]{
  \let\oldtheassumption\theassumption
  \renewcommand{\theassumption}{#1}
  \g@addto@macro\endassumption{
    \global\let\theassumption\oldtheassumption}
  }
\makeatother

\newtheorem{claim}{Claim}
\renewcommand{\theclaim}{\arabic{claim}}
\makeatletter
\newcommand{\setclaimtag}[1]{
  \let\oldtheclaim\theclaim
  \renewcommand{\theclaim}{#1}
  \g@addto@macro\endclaim{
    \global\let\theclaim\oldtheclaim}
  }
\makeatother

\begin{document}


\newcommand{\orcid}[1]{\href{https://orcid.org/#1}{\textcolor[HTML]{A6CE39}{\aiOrcid}}}
\def\bibsection{\section*{\refname}} 

\newcommand{\pbref}[1]{\ref{#1} (\nameref*{#1})}
   
\def\({\big(}
\def\){\big)}

\newcommand{\tn}{\textnormal}
\newcommand{\ds}{\displaystyle}
\newcommand{\dsfrac}[2]{\displaystyle{\frac{#1}{#2}}}

\newcommand{\boplus}{\textstyle{\bigoplus}}
\newcommand{\botimes}{\textstyle{\bigotimes}}
\newcommand{\bcup}{\textstyle{\bigcup}}
\newcommand{\bsqcup}{\textstyle{\bigsqcup}}
\newcommand{\bcap}{\textstyle{\bigcap}}

\newcommand{\struct}{\mc{S}}
\newcommand{\kind}{\mc{K}}

\newcommand{\dddots}{\rotatebox[origin=t]{135}{$\cdots$}}

\newcommand{\statespace}{\mathcal{S}}
\newcommand{\hilbert}{\mathcal{H}}
\newcommand{\vectorspace}{\mathcal{V}}
\newcommand{\mc}[1]{\mathcal{#1}}
\newcommand{\mf}[1]{\mathfrak{#1}}
\newcommand{\mb}[1]{\mathbf{#1}}
\newcommand{\dU}{\wh{\mc{U}}}

\newcommand{\wh}[1]{\widehat{#1}}
\newcommand{\dwh}[1]{\wh{\rule{0ex}{1.3ex}\smash{\wh{\hfill{#1}\,}}}}

\newcommand{\wt}[1]{\widetilde{#1}}
\newcommand{\wht}[1]{\widehat{\widetilde{#1}}}
\newcommand{\on}[1]{\operatorname{#1}}

\newcommand{\qmU}{$\mathscr{U}$}
\newcommand{\qmR}{$\mathscr{R}$}
\newcommand{\qmUR}{$\mathscr{UR}$}
\newcommand{\qmDR}{$\mathscr{DR}$}

\newcommand{\R}{\mathbb{R}}
\newcommand{\C}{\mathbb{C}}
\newcommand{\Z}{\mathbb{Z}}
\newcommand{\K}{\mathbb{K}}
\newcommand{\N}{\mathbb{N}}
\newcommand{\Prj}{\mathcal{P}}
\newcommand{\abs}[1]{\left\lvert#1\right\rvert}

\newcommand{\de}{\operatorname{d}}
\newcommand{\im}{\operatorname{Im}}

\newcommand{\dof}{d.o.f.\xspace}
\newcommand{\dofs}{d.o.f.s\xspace}

\newcommand{\ie}{\textit{i.e.}\ }
\newcommand{\vs}{\textit{vs.}\ }
\newcommand{\eg}{\textit{e.g.}\ }
\newcommand{\cf}{\textit{cf.}\ }
\newcommand{\etc}{\textit{etc}}

\newcommand{\Span}{\tn{span}}
\newcommand{\pde}{PDE}
\newcommand{\U}{\tn{U}}
\newcommand{\SU}{\tn{SU}}
\newcommand{\GL}{\tn{GL}}

\newcommand{\schrod}{Schr\"odinger}
\newcommand{\vonneum}{Liouville-von Neumann}
\newcommand{\ks}{Kochen-Specker}
\newcommand{\leggarg}{Leggett-Garg}
\newcommand{\bra}[1]{\langle#1|}
\newcommand{\ket}[1]{|#1\rangle}
\newcommand{\kett}[1]{|\!\!|#1\rangle\!\!\rangle}
\newcommand{\proj}[1]{\ket{#1}\bra{#1}}
\newcommand{\braket}[2]{\langle#1|#2\rangle}
\newcommand{\ketbra}[2]{|#1\rangle\langle#2|}
\newcommand{\expectation}[1]{\langle#1\rangle}
\newcommand{\Herm}{\tn{Herm}}
\newcommand{\prodHS}[2]{\(#1,#2\)_{\tn{HS}}}
\newcommand{\Sym}[1]{\tn{Sym}_{#1}}
\newcommand{\meanvalue}[2]{\langle{#1}\rangle_{#2}}
\newcommand{\Prob}{\tn{Prob}}
\newcommand{\kjj}[3]{#1\!:\!#2,#3}
\newcommand{\kj}[2]{#1,#2}
\newcommand{\JK}{\mf{j}}
\newcommand{\obs}[1]{\mathsf{#1}}
\newcommand{\uop}[1]{\mathbf{#1}}

\newcommand{\weightU}[5]{\big[{#2}{}_{#3}\overset{#1}{\rightarrow}{#4}{}_{#5}\big]}
\newcommand{\weightUT}[8]{\big[{#3}{}_{#4}\overset{#1}{\rightarrow}{#5}{}_{#6}\overset{#2}{\rightarrow}{#7}{}_{#8}\big]}
\newcommand{\weight}[4]{\weightU{}{#1}{#2}{#3}{#4}}
\newcommand{\weightT}[6]{\weightUT{}{}{#1}{#2}{#3}{#4}{#5}{#6}}

\newcommand{\btimes}{\boxtimes}
\newcommand{\btimess}{{\boxtimes_s}}

\newcommand{\h}{\mathbf{(2\pi\hbar)}}
\newcommand{\x}{\mathbf{x}}
\newcommand{\xThree}{\boldsymbol{x}}
\newcommand{\z}{\mathbf{z}}
\newcommand{\q}{\mathbf{q}}
\newcommand{\p}{\mathbf{p}}
\newcommand{\0}{\mathbf{0}}
\newcommand{\annih}{\widehat{\mathbf{a}}}

\newcommand{\cs}{\mathscr{C}}
\newcommand{\ps}{\mathscr{P}}
\newcommand{\xhat}{\widehat{\x}}
\newcommand{\phat}{\widehat{\mathbf{p}}}
\newcommand{\fqproj}[1]{\Pi_{#1}}
\newcommand{\cqproj}[1]{\wh{\Pi}_{#1}}
\newcommand{\cproj}[1]{\wh{\Pi}^{\perp}_{#1}}

\newcommand{\M}{\mathbb{E}_3}
\newcommand{\D}{\mathbf{D}}
\newcommand{\dn}{\tn{d}}
\newcommand{\db}{\mathbf{d}}
\newcommand{\n}{\mathbf{n}}
\newcommand{\m}{\mathbf{m}}
\newcommand{\V}[1]{\mathbb{V}_{#1}}
\newcommand{\F}[1]{\mathcal{F}_{#1}}
\newcommand{\Fvacuumfield}{\widetilde{\mathcal{F}}^0}
\newcommand{\nD}[1]{|{#1}|}
\newcommand{\Lin}{\mathcal{L}}
\newcommand{\End}{\tn{End}}
\newcommand{\vbundle}[4]{{#1}\to {#2} \stackrel{\pi_{#3}}{\to} {#4}}
\newcommand{\vbundlex}[1]{\vbundle{V_{#1}}{E_{#1}}{#1}{M_{#1}}}
\newcommand{\rep}{\rho_{\scriptscriptstyle\btimes}}

\newcommand{\intl}[1]{\int\limits_{#1}}

\newcommand{\moyalBracket}[1]{\{\mskip-5mu\{#1\}\mskip-5mu\}}

\newcommand{\Hint}{H_{\tn{int}}}

\newcommand{\quot}[1]{``#1''}

\def\sref #1{\S\ref{#1}}

\newcommand{\dBB}{de Broglie--Bohm}
\newcommand{\dBBt}{{\dBB} theory}
\newcommand{\pwt}{pilot-wave theory}
\newcommand{\PWT}{PWT}
\newcommand{\NRQM}{{\textbf{NRQM}}}

\newcommand{\image}[3]{
\begin{center}
\begin{figure}[!ht]
\includegraphics[width=#2\textwidth]{#1}
\caption{\small{\label{#1}#3}}
\end{figure}
\end{center}
\vspace{-0.40in}
}

\clearpage

\title[Can we accurately read or write quantum data?]{Can we accurately read or write quantum data?}

\author{Ovidiu Cristinel Stoica\ \orcidlink{0000-0002-2765-1562}}
\address{
 Dept. of Theoretical Physics, NIPNE---HH, Bucharest, Romania. \\
	Email: \href{mailto:cristi.stoica@theory.nipne.ro}{cristi.stoica@theory.nipne.ro},  \href{mailto:holotronix@gmail.com}{holotronix@gmail.com}
	}%

\date{\today}

\begin{abstract}
Applications of quantum mechanics rely on the accuracy of reading and writing data. This requires accurate measurements and preparations of the quantum states. I show that accurate measurements and preparations are impossible if the total Hamiltonian is bounded from below (as thought to be in our universe). This result invites a reevaluation of the limitations of quantum control, quantum computing, and other quantum technologies dependent on the accuracy of quantum preparations and measurements, and maybe of the assumption that the Hamiltonian is bounded from below.\end{abstract}

\vspace{2pc}

\maketitle

\section{Introduction}
\label{s:intro}

Quantum measurements are extremely important for all domains of the quantum, from its foundations to experiments and technological applications.

When confronted with technological limitations, for example due to errors and decoherence, we can easily get trapped in a \emph{bottom-up} paradigm. The bottom-up paradigm is based on the assumption that if we solve the problems at a small scale, we can generalize the solution at a large scale. That is, it assumes that if we can build perfect quantum gates and isolate them from the environment, or if we apply quantum error correction codes, it becomes a matter of time and financial resources to apply them to build any quantum circuit and any quantum technology we can imagine on paper.

The sky seems to be the limit of quantum technologies.
But up-there in the sky is Spinoza's God, the laws of physics. Every quantum technology we have in mind has to be realized in this universe, and has to work according to the laws of physics of this universe. The physical laws may impose \emph{top-down} constraints on what can be achieved.
And if there are such constraints, they can't be avoided, no matter how great researchers and engineers we have and how much time and money we invest.

We can classify the top-down problems that need to be solved in two major problems:
\begin{problem}[Dynamics]
\label{problem:invariance}
The unitary operators that we want to build as quantum gates or by combining them should be realized as restrictions of the {\schrod} unitary evolution to invariant subspaces of the total Hilbert space of our world. Since we are given only this world for this, we have to investigate what unitary operators can be realized in our world, and not merely \emph{in abstracto}.
\end{problem}

\begin{problem}[Boundary conditions]
\label{problem:accuracy}
Because the technologies based on quantum information rely on reading and writing data, the accuracy of quantum measurements, either to extract information, or to prepare the quantum states, should be investigated.
\end{problem}

Therefore, the possible limitations of quantum measurements, and the possibility to accurately measure or prepare a quantum state, are of particular interest\footnote{
The term ``accuracy'' is the technical term used in the literature about quantum measurements, and it can be understood as the precision with which the obtained value corresponds to the actual eigenstate of the observed system. Another term is ``repeatability'', which means that if the measurement is repeated, the same result obtains, or, in other words, the measurement doesn't disturb the observed system {\cite{Busch2009NoInformationWithoutDisturbance}}.}.
Wigner \cite{wigner1952MessungQMOperatoren,Wigner1952MessungQMOperatorenPBusch2010EnTranslation}, followed by Araki and Yanase \cite{ArakiYanase1960MeasurementofQMOperators}, realized that additive conservation laws restrict what sharp observables can be accurately measured without being disturbed.
Wigner showed that the ideal measurement of the spin along an axis is prevented by the conservation of the total angular momentum along an orthogonal axis. From this result, known as the \emph{Wigner-Araki-Yanase} (WAY) theorem \cite{LoveridgeBusch2011MeasurementofQMOperatorsRevisited}, it seems that we have to choose between inaccuracy and disturbing the system even when it is in an eigenstate of the observable.

But, in the same paper, Wigner also showed that we can avoid this problem if the measuring device is infinitely large.

Even if the same measurement cannot be both accurate and non-disturbing, for some applications it seems enough to achieve accuracy and non-disturbance in separate measurements.
We can write (or prepare) the states using non-disturbing measurements, and read (or measure) the resulting states using accurate measurements.

While Wigner concluded in \cite{Wigner1952MessungQMOperatorenPBusch2010EnTranslation} that accurate but non-repeatable (non-disturbing) measurements are possible even without the requirement that the apparatus is very large, Ohira and Pearle showed that this would violate the Yanase condition \cite{OhiraPearle1988PerfectDisturbingMeasurements}.
The \emph{Yanase condition} \cite{Yanase1961OptimalMeasuringApparatus} requires the probe observable to commute with the conserved quantity, and it is necessary if we want to observe the result of the measurement.

Despite Wigner's well known interest in symmetries, he seemed to use the existence of a conservation law as a mere signature of unitarity. But there seems to be a deeper lesson about symmetry encoded in the WAY theorem. A deep relation between symmetries, measurements, and quantum reference frames was clarified by Loveridge, Miyadera, and Busch \cite{LoveridgeMiyaderaBusch2018SymmetryReferenceFramesRelationalQuantitiesInQuantumMechanics}.
However, the constraints on measurements extend beyond the direct applications of symmetries, as we can see from Ozawa's result that measurements of non-discrete observables are not repeatable \cite{Ozawa1984QuantumMeasuringProcessesOfContinuousObservables}, and from Tukiainen's generalization of the WAY theorem beyond conservation laws \cite{Tukiainen2017WignerArakiYanaseTheoremBeyondConservationLaws}.
For an up to date review of results about limitations of quantum measurements see Busch \cite{Busch2009NoInformationWithoutDisturbance}.
The WAY theorem was later extended to the case when the conserved quantity has a continuous and unbounded spectrum \cite{KuramochiTajima2023WAYTheoremForContinuousAndUnboundedConservedObservables}.
The measurability and disturbance parts of the WAY theorem were also extended to unsharp observables (positive operator-valued measures), in the presence of bounded conserved quantities, and necessary conditions for either accurate (in an unsharp, probabilistic sense) or non-disturbing measurements can be achieved \cite{MohammadyMiyaderaLoveridge2023MeasurementDisturbanceAndConservationLawsInQuantumMechanics}.
But the full scope of the WAY theorem and its generalizations is still being charted.

Another limitation comes from the fact that the Third Law of Thermodynamic prohibits the preparation of pure states \cite{SchulmanMorWeinstein2005PhysicalLimitsOfHeatBathAlgorithmicCooling,MasanesOppenheim2017AGeneralDerivationAndQuantificationOfTheThirdLawOfThermodynamics}.
Recently it has been shown that ideal projective measurements require infinite resource costs \cite{GuryanovaFriisHuber2020IdealProjectiveMeasurementsHaveInfiniteResourceCosts}. Without a proper preparation of the pointer state in the ``ready'' state, the statistics of the results will be biased. Because of the Third Law of Thermodynamics, infinite resource costs are needed for the preparation of the pointer in a suitable state, and also to ensure that the interaction between the observed system and the pointer achieves the correct correlation. 
In \cite{Mohammady2023QuantumMeasurementsConstrainedByTheThirdLawOfThermodynamics} ways to avoid this limitation were explored, suggesting that sufficiently unsharp observables may avoid it.

Ozawa warned about the possible top-down limitations imposed by the WAY theorem on quantum computing in \cite{Ozawa2002ConservativeQuantumComputing,Ozawa2003UncertaintyPrincipleForQuantumInstrumentsAndComputing}.
Lidar replied that it is already taken into account that quantum computing should be implemented on invariant subspaces of the Hilbert space \cite{Lidar2003CommentOnConservativeQuantumComputing}.
More investigations of the limits imposed by conservation laws on quantum computing were reported in \cite{KarasawaOzawa2007Conservation-law-inducedQuantumLimitsForPhysicalRealizationsOfTheQuantumNOTGate,KarasawaGeaBanaclocheOzawa2009GateFidelityOfArbitrarySingleQubitGatesConstrainedByConservationLaws}.

An apparently unrelated factor is the spectrum of the Hamiltonian, which is axiomatically considered to be bounded from below, {\ie} that all energy eigenvalues are larger than some fixed real value \cite{Pauli1980GeneralPrinciplesOfQuantumMechanics}.
A reason often invoked is that otherwise a system will presumably decay to arbitrarily low energies \cite{Dirac1930ATheoryOfElectronsAndProtons}.

Here I show that if the Hamiltonian is bounded from below, no quantum measurement of a sharp observable, and no state preparation, can be accurate.
The main result is proved for pure states in Section \sref{s:theorem}, and generalized to mixed states in Section \sref{s:mixed}. Section \sref{s:discussion} concludes the article with a brief discussion of the implications of these results.

\section{The main result}
\label{s:theorem}

We consider the quantum measurement of an observable represented by a self-adjoint operator $\obs{A}$.
The observed system $S$ is assumed to be initially in a state represented by the unit vector $\psi(0)$ from a Hilbert space $\hilbert_{S}$.
The initial, ``ready'' state of the measuring device $M$ is denoted by a unit vector $\varphi_{\varnothing}$ from a Hilbert space $\hilbert_{M}$.

Let the \emph{measured observable} $\obs{A}$ be, in terms of projectors $\obs{P}_\lambda$ adding up to the identity on $\hilbert_{S}$,
\begin{equation}
\label{eq:observable-S-spectral}
\obs{A}=\sum_\lambda\lambda\obs{P}_\lambda.
\end{equation}

Let the \emph{pointer observable} $\obs{Z}$ of the measuring device be, in terms of projectors $\Pi_\xi$ adding up to the identity on $\hilbert_{M}$,
\begin{equation}
\label{eq:observable-M-spectral}
\obs{Z}=\sum_\xi\xi\Pi_\xi.
\end{equation}

The observables $\obs{A}$ and $\obs{Z}$ may have continuous spectra, but for simplicity we will use the notations \eqref{eq:observable-S-spectral} and \eqref{eq:observable-M-spectral}.
In any case, we can decompose the continuous spectra into discrete regions and be interested only in discrete differences, since in practice we discretize the results of the measurements anyway.
But the proofs from this article can be applied to continuous spectra as well.

We can assume without loss of generality that the spectrum of the observable $\obs{Z}$ contains the spectrum of the observable $\obs{A}$, and an extra eigenvalue for the ``ready'' state $\varphi_{\varnothing}$.

We make the following assumption about the total Hamiltonian:
\begin{assumption}[No negative energy]
\label{assumption:non-negativity}
The Hamiltonian of the combined system consisting of the measuring device and the observed system is bounded from below.
\end{assumption}

We also make the following assumption about the measurement:
\begin{assumption}[Isolation]
\label{assumption:isolated}
The combined system consisting of the measuring device and the observed system is isolated for the duration of the measurement.
\end{assumption}

The evolution of the combined system $S+M$ is governed by the unitary operator $\uop{U}_t=e^{-\frac{i}{\hbar}t\obs{H}}$.
We consider interactions of the form
\begin{equation}
\label{eq:premeasurement}
\uop{U}_T(\psi(0)\otimes\underbrace{\varphi_{\varnothing}(0)}_{\mathclap{\in\Pi_{\varnothing}\hilbert_{M}}})=\sum_\lambda c_\lambda\psi_\lambda(T)\otimes \underbrace{\varphi_\lambda(T)}_{\in\Pi_{\lambda}\hilbert_{M}},
\end{equation}
where $\varphi_{\varnothing}(0)$ and each $\varphi_\lambda(T)$ are eigenstates of the pointer observable $\obs{Z}$\footnote{If the operator {$\obs{Z}$} is nondegenerate, it may seem redundant to  include the dependency of time in the notation when we already labeled the eigenstates with the eigenvalue, so it may seem that we could just write, for example, {$\varphi_\lambda$} instead of {$\varphi_\lambda(T)$}. But we do it anyway because we want to allow degeneracy for maximal generality, and because we want to indicate the time.}. Here, {$\psi_\lambda(T)$} is not necessarily an eigenvalue of {$\obs{A}$}, the label {$\lambda$} being used for summation.

We are interested in two particular kinds of processes, measurements and preparations.

\begin{definition}
\label{def:accurate-measurement}
An \emph{accurate measurement} happens if, assuming that the pointer was in a ``ready'' state {$\varphi_{\varnothing}(0)$} at $t=0$, $\uop{U}_T$ maps any separable state $\psi_\lambda(0)\otimes\varphi_{\varnothing}(0)$, where $\psi_\lambda(0)$ is an eigenstate of $\obs{A}$, to another separable state $\psi_\lambda(T)\otimes \varphi_\lambda(T)$,
\begin{equation}
\label{eq:calibration-measurement}
\uop{U}_T\(\underbrace{\psi_\lambda(0)}_{\in\obs{P}_{\lambda}\hilbert_{S}}\otimes\varphi_{\varnothing}(0)\)\underset{\substack{\phantom{O}\\\Longleftrightarrow}}{=}\psi_\lambda(T)\otimes \underbrace{\varphi_\lambda(T)}_{\in\obs{\Pi}_{\lambda}\hilbert_{M}},
\end{equation}
where $\varphi_\lambda(T)$ is an eigenstate of the pointer observable $\obs{Z}$, and, from the pointer indicating the value $\lambda$ of the pointer observable $\obs{Z}$, we can infer that the observed system was, at $t=0$, in an eigenstate of $\obs{A}$ with eigenvalue $\lambda$.
This is not assumed to be true if the observed system was, initially, not in an eigenstate of the observable or if the pointer was not in a ``ready'' state at $t=0$.
\end{definition}

\begin{definition}
\label{def:accurate-preparation}
An \emph{accurate preparation} happens if, assuming that the pointer was in a ``ready'' state {$\varphi_{\varnothing}(0)$} at $t=0$, from the pointer indicating the value $\lambda$ at the end of the measurement $t=T$, we can infer that $\psi_\lambda(T)$ is an eigenstate of $\obs{A}$ with eigenvalue $\lambda$,
\begin{equation}
\label{eq:calibration-preparation}
\uop{U}_T\(\psi(0)\otimes\varphi_{\varnothing}(0)\)=\sum_\lambda c_\lambda\underbrace{\psi_\lambda(T)}_{\in\obs{P}_{\lambda}\hilbert_{S}}\underset{\substack{\phantom{O}\\\Longleftarrow}}{\otimes} \underbrace{\varphi_\lambda(T)}_{\in\obs{\Pi}_{\lambda}\hilbert_{M}}.
\end{equation}
\end{definition}

At time $t=T$ we retain from the superposition \eqref{eq:premeasurement} a single outcome $\psi_\lambda(T)\otimes \varphi_\lambda(T)$, by invoking our favorite recipe, for example the Projection Postulate or an alternative proposal based on decoherence like the many-worlds interpretation or the pilot-wave theory,
\begin{equation}
\label{eq:collapse}
\sum_\lambda c_\lambda\psi_\lambda(T)\otimes \varphi_\lambda(T)\mapsto \psi_{\lambda_{\tn{o}}}(T)\otimes \varphi_{\lambda_{\tn{o}}}(T),
\end{equation}
where {$\lambda_{\tn{o}}$} is a particular outcome of the measurement or the preparation.

In both cases, we also make the following assumption, necessary for the readability of the pointer state:
\begin{assumption}[Pointer persistence]
\label{assumption:stable}
There is a time interval $[T,T']$, $T'>T$, so that, if the total system $S+M$ remains isolated for $t\in[T,T']$, the pointer states $\varphi_\lambda(t)$ remain eigenstates of $\obs{Z}$ for each eigenvalue $\lambda$, that is,
\begin{equation}
\label{eq:pointer:stable}
\varphi_\lambda(t) \in \Pi_{\lambda}\hilbert_{M}
\end{equation}
for all $t\in[T,T']$ and any $\lambda$, where $\Pi_{\lambda}\hilbert_{M}$ is the eigenspace of $\obs{Z}$ corresponding to the eigenvalue $\lambda$ of $\obs{Z}$, {\cf} equation \eqref{eq:observable-M-spectral}.
\end{assumption}

\begin{proof}[Justification]
Note that {$\obs{Z}$} can be degenerate (and it usually is highly degenerate, because the pointer is a macroscopic object), and {$\varphi_\lambda(t)$} can change in time, but, for $t\in[T,T']$, {$\varphi_\lambda(t)$} is assumed to remain confined to the eigenspace {$\Pi_{\lambda}\hilbert_{M}$} of {$\obs{Z}$}, otherwise it would be impossible to know the result of the measurement. Whatever kind of quantum measurement we may have in mind, the experimentalist needs to be able to read the result. The outcome is recorded temporarily somewhere, on a photographic plate, in the position of a pointer and so on. It would be unrealistic to assume that it is recorded forever, but at least for a finite time interval $[T,T']$, the experimentalist should be able to read the result.
And the result is recorded in the eigenvalue of the pointer's state vector, which is an eigenvector of the pointer observable. This minimal requirement absolutely necessary for experiments to be possible is expressed in Assumption~{\ref{assumption:stable}}.
\end{proof}

The Definitions~{\ref{def:accurate-measurement}} and {\ref{def:accurate-preparation}} allow treating the measurement and the preparation as particular cases of equation~{\eqref{eq:premeasurement}}, where the pointer indicates the state of the observed system at time $t=0$ (measurements) and $t=T$ (preparations).
Both of these cases will be covered by the proof of the following result:
\begin{theorem}
\label{thm:no-accuracy}
If the Hamiltonian is bounded from below, accurate measurements and preparations are impossible.
\end{theorem}
\begin{proof}
We will need the Lemma from \cite{HegerfeldtRuijsenaars1980RemarksOnCausalityLocalizationAndSpreadingOfWavePackets}, p. 378:
\begin{lemma}[Hegerfeldt \& Ruijsenaars]
\label{lemma:HegerfeldtRuijsenaars}
Let $\obs{H}$ be a self-adjoint operator on a Hilbert space $\hilbert$, $\hilbert'$ a closed subspace of $\hilbert$, and $\Psi_0\in\hilbert$. If $\obs{H}$ is bounded from below and $e^{-\frac{i}{\hbar}t\obs{H}}\Psi_0\in\hilbert'$ for all values of $t$ in an open interval $(T,T')$, then $e^{-\frac{i}{\hbar}t\obs{H}}\Psi_0\in\hilbert'$ for all $t\in\R$.
\qed
\end{lemma}

We apply Lemma \ref{lemma:HegerfeldtRuijsenaars} to the following case (with the notations from the Lemma):
\begin{enumerate}
	\item $\hilbert=\hilbert_{S}\otimes\hilbert_{M}$.
	\item $\Psi_0=\uop{U}_{T}^{-1}\(\psi_\lambda(T)\otimes \varphi_\lambda(T)\)$, where {$\varphi_\lambda(T)$} is an eigenstate of the pointer observable.
	\item $\hilbert'=\hilbert_{S}\otimes\Pi_{\lambda}\hilbert_{M}$.
\end{enumerate}

The conditions~{\eqref{eq:calibration-measurement}}, respectively~{\eqref{eq:calibration-preparation}}, apply, in particular, to initial states that evolve into a definite value for the pointer without having to appeal to the projection from equation~{\eqref{eq:collapse}}. In the case of a measurement, this happens if
\begin{equation}
\label{eq:initial-state-measurement}
\Psi_0=\psi(0)\otimes\varphi_{\varnothing}(0),
\end{equation}
where {$\psi(0)$} is an eigenvector {$\psi_{\lambda}(0)$} of {$\obs{A}$} for some eigenvalue {$\lambda$}. In the case of a preparation, this happens if {$\Psi_0$} is an eigenvector of the operator resulting by evolving back to $t=0$ the operator $\obs{P}_{\lambda}\hilbert_{S}\otimes\obs{\Pi}_{\lambda}\hilbert_{M}$. In both cases, the pointer is assumed to start in a ``ready'' state. Since both measurements and preparations should work as well for these situations in which the evolution is strictly unitary, we can use them to derive a contradiction.

From Assumption \ref{assumption:stable}, $\uop{U}_t\Psi_0\in\hilbert'$ for all $t\in(T,T')$, so the hypothesis of Lemma \ref{lemma:HegerfeldtRuijsenaars} is satisfied.
Then, its conclusion is that $\uop{U}_t\Psi_0\in\hilbert'$ for all $t\in\R$ and all eigenvalues {$\lambda$}.
In particular, for $t=0$, since $\uop{U}_0\(\Psi_0\)=\Psi_0$, we obtain
\begin{equation}
\label{eq:wrong-initial-state}
\Psi_0\in \hilbert_{S}\otimes\Pi_{\lambda}\hilbert_{M}.
\end{equation}
From this,
\begin{equation}
\label{eq:wrong-ready-state}
\varphi_{\varnothing}(0)\in\Pi_{\lambda}\hilbert_{M},
\end{equation}
in contradiction with the assumption that the ``ready'' pointer state $\varphi_{\varnothing}(0)\in\Pi_{\varnothing}\hilbert_{M}$, or, in general, that $\varphi_{\varnothing}(0)$ cannot be an eigenstate of $\obs{Z}$ for the eigenvalue $\lambda$!

Therefore, if the Hamiltonian is bounded from below, equation \eqref{eq:calibration-measurement} cannot be satisfied, so accurate measurements are impossible.
Similarly, equation \eqref{eq:calibration-preparation} cannot be satisfied, so accurate preparations are impossible too.
\end{proof}

We will deal more with the general cases when the observed system and the measuring device are entangled or even mixed states in Section~{\sref{s:mixed}}.

\begin{remark}
\label{rem:thm-explained}
The proof is based on extracting conditions for the ``ready'' pointer state from condition {\eqref{eq:calibration-measurement}}, respectively {\eqref{eq:calibration-preparation}}, applied to different values of {$\lambda$}. For the eigenstates, the evolution is unitary, without projection,  at least between $t=0$ and $t=T'$, because the observed system already is in an eigenstate of the observable.
As seen in the proof, these conditions turn out to require, for each eigenvalue {$\lambda$}, that the ``ready'' pointer state has to be the pointer state corresponding to {$\lambda$}, which is a contradiction. 

Since conditions {\eqref{eq:calibration-measurement}} and {\eqref{eq:calibration-preparation}} are assumed to be satisfied by the measuring device in standard Quantum Mechanics, they are assumed in all interpretations, even those that include the Projection Postulate, because this Postulate is invoked only if the measurement process leads to a superposition of different outcomes.
\end{remark}

\section{Methods, explanations, and physical interpretation}
\label{s:methods}

The method used in the proof relies on obtaining a contradiction from the assumption that the Hamiltonian is bounded from below, so it is a proof by \emph{reductio ad absurdum}. Theorem~{\ref{thm:no-accuracy}} also assumes that the measurement takes place in isolation (Assumption~{\ref{assumption:isolated}}), and that the value indicated by the pointer persists for at least a finite duration (Assumption~{\ref{assumption:stable}}).
The model of measurement is the most general that assumes that the observed system is initially separated from the measuring device and that the observable is sharp.
However, the condition of separability is dropped in the generalization from Theorem~{\ref{thm:no-accuracy-mixed}}, Section~{\sref{s:mixed}}.

The proof starts with the conditions {\eqref{eq:calibration-measurement}} and respectively {\eqref{eq:calibration-preparation}}, sometimes called in the literature \emph{calibration conditions}, \cite{BuschLahtiPellonpaaYlinen2016QuantumMeasurement,Busch2009NoInformationWithoutDisturbance,BuschJaeger2010UnsharpQuantumReality}.
These conditions have to be satisfied regardless of the interpretation of Quantum Mechanics, and don't require projection or collapse, because the state is already an eigenstate.
The proof shows that, if the observed system is in an eigenstate for the eigenvalue $\lambda$ of $\obs{A}$, the assumption that the Hamiltonian is bounded from below (Assumption~{\ref{assumption:non-negativity}}), combined with the other assumptions, implies that the pointer's ``ready'' state $\varphi_{\varnothing}(0)$ has to be an eigenvalue for the eigenvalue $\lambda$ for $\obs{Z}$.
Consider two distinct observables $\lambda\neq\lambda'$.
Then, the proof of Theorem~{\ref{thm:no-accuracy}} shows that $\varphi_{\varnothing}(0)$ has to be an eigenvector for both $\lambda$ and $\lambda'$. But this is possible only if $\lambda=\lambda'$, in which case there is no measurement, or if $\varphi_{\varnothing}=0$, but this vector can't represent a pointer state because it is not a unit vector. The pointer states $\varphi_{\varnothing}(0)$, $\varphi_{\lambda}$, and $\varphi_{\lambda'}$ have to be mutually orthogonal, as in Figure~{\ref{fig:contradiction}}.

\begin{figure}[h!]
\centering
\tdplotsetmaincoords{60}{120} 
\begin{tikzpicture} [scale=3, tdplot_main_coords, axis/.style={->,blue,thick}, 
vector/.style={-stealth,red,ultra thick}]

\coordinate (O) at (0,0,0);

\draw[axis] (0,0,0) -- (1,0,0) node[anchor=north east]{$\varnothing$};
\draw[axis] (0,0,0) -- (0,1,0) node[anchor=north west]{$\lambda$};
\draw[axis] (0,0,0) -- (0,0,1) node[anchor=south]{$\lambda'$};

\draw[vector] (O) -- (0.7,0,0) node[anchor=east]{$\varphi_{\varnothing}$};
\draw[vector] (O) -- (0,0.7,0) node[anchor=south]{$\varphi_{\lambda}$};
\draw[vector] (O) -- (0,0,0.7) node[anchor=east]{$\varphi_{\lambda'}$};

\draw[blue] (0,0,0.15) -- (0,0.15,0.15) -- (0,0.15,0) -- (0.15,0.15,0) -- (0.15,0,0) -- (0.15,0,0.15) -- (0,0,0.15);

\end{tikzpicture}
\caption{The contradiction derived in the proof is that the pointer's ready state $\varphi_{\varnothing}$ has to be an eigenstate for both $\lambda$ and $\lambda'$, while in fact $\varphi_{\varnothing}$, $\varphi_{\lambda}$, and $\varphi_{\lambda'}$ have to be mutually orthogonal.}
\label{fig:contradiction}
\end{figure}

Perhaps the least intuitive part of the proof is Lemma~{\ref{lemma:HegerfeldtRuijsenaars}}, so let's give here its proof, following \cite{HegerfeldtRuijsenaars1980RemarksOnCausalityLocalizationAndSpreadingOfWavePackets}.
\begin{proof}[Proof of Lemma~\ref{lemma:HegerfeldtRuijsenaars}]
Let $\Phi$ be an arbitrary vector orthogonal to $\hilbert'$. Let $z=t+iy$, with $y\leq 0$. Then, because $\obs{H}$ is bounded from below, the function $g(z):=\braket{\Phi}{e^{-i\obs{H}z}\Psi}$ is analytic for $y<0$ (see for example \cite{StreaterWightman2000PCTSpinAtatisticsAndAllThat}). The function $g(z)$ is continuous for $y\leq 0$ and, because $\Psi(t)\in\hilbert'$ for $t\in(T,T')$, $g(z)$ vanishes for $z=t\in(T,T')$. For the next step of the proof we apply a theorem in complex analysis, called the \emph{Schwarz reflection principle} (see for example \cite{Simon2015BasicComplexAnalysis,StreaterWightman2000PCTSpinAtatisticsAndAllThat}), according to which if the function $g(z)$ defined on $\{z=t+iy\in\C|y\leq 0\}$ is analyting for $y<0$, continuous for $y\leq 0$, and real for $z=t\in(T,T')$, it can be extended analytically on the entire complex plane $\C$ by making $g(t-iy):=g(t+iy)^\ast$ for all $y> 0$. Since $g(z)$ is analytic and vanishes on $(T,T')$, it follows that it vanishes everywhere, in particular for all $t\in\R$. This implies that $\Phi\perp\obs{U}_t$ for all $t\in\R$, so $e^{-\frac{i}{\hbar}t\obs{H}}\Psi_0\in\hilbert'$ for all $t\in\R$.
\end{proof}

\section{Generalization to mixed states}
\label{s:mixed}

We may hope to get around Theorem \ref{thm:no-accuracy} by relaxing the defining conditions of measurements and preparations.
For example, we can allow the observed system and the pointer to be in mixed states.
Most generally, we can describe their combined state by a density operator $\rho(t)$, and not by a state vector $\Psi(t)$.
Then, if the pointer state is initially in a mixed ``ready'' state, the initial state of the combined system satisfies the condition
\begin{equation}
\label{eq:mixed-ready-pointer-initial}
\tr_{S}\rho(0)=\Pi_{\varnothing}\(\tr_{S}\rho(0)\)\Pi_{\varnothing}^\dagger,
\end{equation}
where $\tr_{S}\(\rho(0)\)$ is the reduced density operator obtained by partial trace over the observed system.
The meaning of equation \eqref{eq:mixed-ready-pointer-initial} is that the reduced density operator representing the pointer vanishes outside of the eigenspace representing the ``ready'' state of the pointer.

Condition \eqref{eq:mixed-ready-pointer-initial} may be too relaxed, because it allows the observed system and the measuring device to be entangled from the start. If they are not entangled, $\rho(0)$ has the form $\rho(0)=\rho(0)_{S}\otimes\rho(0)_{M}$. Then already $\tr_{S}\rho(0)=\tr_{S}\(\rho(0)_{S}\otimes\rho(0)_{M}\)=\rho(0)_{M}$ and condition \eqref{eq:mixed-ready-pointer-initial} becomes $\Pi_{\varnothing}\rho(0)_{M}\Pi_{\varnothing}^\dagger=\rho(0)_{M}$.
Therefore, condition \eqref{eq:mixed-ready-pointer-initial} is more relaxed than normally needed, but we will allow it for the sake of full generality.

After the interaction between the measuring device and the observed system, the Projection Postulate (or a decoherence-based approach) leaves us with a mixture of states containing pointer eigenstates of $\obs{Z}$ for only one of the eigenvalues $\lambda$.
That is, it leaves us with the density operator $\(\obs{I}_{S}\otimes\Pi_{\lambda}\)\rho(T)\(\obs{I}_{S}\otimes\Pi_{\lambda}\)^\dagger$.

In the case of an accurate measurement, if at the beginning of the interaction the pointer state satisfied condition \eqref{eq:mixed-ready-pointer-initial} and the observed system was a mixture of eigenstates for the eigenvalue $\lambda$, that is, it satisfied
\begin{equation}
\label{eq:mixed-measurement-observed-initial}
\tr_{M}\rho(0)=\obs{P}_{\lambda}\(\tr_{M}\rho(0)\)\obs{P}_{\lambda}^\dagger,
\end{equation}
then, at the end of the interaction, the total state satisfies
\begin{equation}
\label{eq:mixed-projected-pointer-final-lambda}
\rho(T)=\(\obs{I}_{S}\otimes\Pi_{\lambda}\)\rho(T)\(\obs{I}_{S}\otimes\Pi_{\lambda}\)^\dagger
\end{equation}
and conversely.

In the case of an accurate preparation, if condition \eqref{eq:mixed-projected-pointer-final-lambda} is satisfied ({\ie} the pointer state is a mixture of eigenstates for the eigenvalue $\lambda$ of $\obs{Z}$), the density operator of the observed system at the time $T$ is a mixture of eigenstates of $\obs{A}$ for the eigenvalue $\lambda$, so it has to vanish outside of the eigenspace  of $\obs{A}$ with eigenvalue $\lambda$,
\begin{equation}
\label{eq:mixed-preparation-observed-initial}
\tr_{M}\rho(T)=\obs{P}_{\lambda}\(\tr_{M}\rho(T)\)\obs{P}_{\lambda}^\dagger.
\end{equation}

To accommodate mixed states, we relax condition \eqref{eq:pointer:stable} from Assumption \ref{assumption:stable} to
\begin{equation}
\label{eq:pointer:stable-mixed}
\tr_{S}\rho(t)=\Pi_{\lambda}\(\tr_{S}\rho(t)\)\Pi_{\lambda}^\dagger
\end{equation}
for any $t\in[T,T']$.

These conditions give the most relaxed, the most general characterization of accurate measurements or preparations of mixed states for sharp observables.
If we try to generalize it even more by considering unsharp observables (\emph{positive operator-valued measures} \cite{DaviesLewis1970OperationalApproachToQuantumProbability,BuschLahtiPellonpaaYlinen2016QuantumMeasurement}) instead of the projective observable $\obs{A}$, measurements or preparations will automatically be unable to accurately distinguish states, by construction.

Theorem \ref{thm:no-accuracy} generalizes to mixed states,
\begin{theorem}
\label{thm:no-accuracy-mixed}
If the Hamiltonian is bounded from below, accurate measurements and preparations of mixed states are impossible.
\end{theorem}
\begin{proof}
The bounded linear operators on the Hilbert space $\hilbert$ form a complex vector space $\mc{B}\(\hilbert\)$, endowed with the \emph{Hilbert-Schmidt inner product}, defined for any two Hermitian operators $\obs{B},\obs{C}\in\mc{B}\(\hilbert\)$ by
\begin{equation}
\label{eq:HS}
\prodHS{\obs{B}}{\obs{C}}:=\tr\(\obs{B}^\dagger\obs{C}\).
\end{equation}

This inner product is inherited by the real vector subspace $\Herm\(\hilbert\)$ of bounded Hermitian operators on $\hilbert$.
The density operators form a convex set in $\Herm\(\hilbert\)$.

The unitary evolution operators $\uop{U}_t=e^{-\frac{i}{\hbar}t\obs{H}}$ evolve any Hermitian operator $\obs{B}$ into $\uop{U}_t\obs{B}\uop{U}_t^\dagger$, preserving the Hilbert-Schmidt inner product,
\begin{equation}
\label{eq:unitary-HS}
\prodHS{\obs{B}}{\obs{C}}=\prodHS{\uop{U}_t\obs{B}\uop{U}_t^\dagger}{\uop{U}_t\obs{C}\uop{U}_t^\dagger}.
\end{equation}
This applies in particular to density operators.

We define the following operators on $\hilbert_{S}\otimes\hilbert_{M}$,
\begin{equation}
\label{eq:extended-trace}
\begin{aligned}
\wt{\tr}_{S}(\rho) &:=
\obs{I}_{S}\otimes\tr_{S}(\rho)\\
\wt{\tr}_{M}(\rho) &:=
\tr_{M}(\rho)\otimes\obs{I}_{M}.\\
\end{aligned}
\end{equation}

They are linear operators, because the partial traces are linear maps.
Conditions \eqref{eq:mixed-ready-pointer-initial}--\eqref{eq:pointer:stable-mixed} can be rewritten in terms of the operators $\wt{\tr}_{S}$ and $\wt{\tr}_{M}$ and the projectors
\begin{equation}
\label{eq:extended-proj}
\begin{aligned}
\wt{\obs{P}}_{\lambda}&:=\obs{P}_{\lambda}\otimes\obs{I}_{M}, & & \\
\wt{\Pi}_{\lambda}&:=\obs{I}_{S}\otimes\Pi_{\lambda}, & \wt{\Pi}_{\varnothing}&:=\obs{I}_{S}\otimes\Pi_{\varnothing}, \\
\end{aligned}
\end{equation}
all operating on the total Hilbert space, but also on $\Herm\(\hilbert\)$, for example as in \eqref{eq:mixed-projected-pointer-final-lambda}.

In terms of the operators \eqref{eq:extended-trace} and \eqref{eq:extended-proj}, condition \eqref{eq:mixed-ready-pointer-initial} becomes $\wt{\tr}_{S}\rho(0)=\wt{\Pi}_{\varnothing}\(\wt{\tr}_{S}\rho(0)\)\wt{\Pi}_{\varnothing}^\dagger$.
Since the partial trace $\tr_{S}$ takes density operators on $\hilbert_{S}\otimes\Pi_{\varnothing}\hilbert$ into density operators on $\Pi_{\varnothing}\hilbert$, the operator $\wt{\tr}_{S}$ takes density operators on $\hilbert_{S}\otimes\Pi_{\varnothing}\hilbert$ into Hermitian operators on $\hilbert_{S}\otimes\Pi_{\varnothing}\hilbert$.
Because of this, when condition \eqref{eq:mixed-ready-pointer-initial} is satisfied,
\begin{equation}
\label{eq:mixed-ready-pointer-initial-HS}
\rho(0)=\wt{\Pi}_{\varnothing}\rho(0)\wt{\Pi}_{\varnothing}^\dagger.
\end{equation}

By a similar reasoning, when condition \eqref{eq:pointer:stable-mixed} is satisfied, for all $t\in(T,T')$,
\begin{equation}
\label{eq:pointer:stable-mixed-HS}
\rho(t)=\wt{\Pi}_{\lambda}\rho(t)\wt{\Pi}_{\lambda}^\dagger.
\end{equation}

Taking into account that density operators on $\hilbert$ are also vectors in the vector space $\Herm\(\hilbert\)$, we will apply Lemma \ref{lemma:HegerfeldtRuijsenaars} to the following case.
As a Hilbert space we use the vector space of positive semi-definite trace-class operators with finite Hilbert-Schmidt norm. Let us denote it by $\wt{\hilbert}$. It is a real Hilbert space.
Anticipating the potential objection that Lemma \ref{lemma:HegerfeldtRuijsenaars} was proved for complex, not for real Hilbert spaces, we can take the complexified $\wt{\hilbert}\otimes\C$ instead of $\wt{\hilbert}$, and extend the Hilbert-Schmidt inner product \eqref{eq:HS} to a Hermitian scalar product, obtaining a complex Hilbert space.
Then the unitary evolution operators $\uop{U}_t$, extended by linearity to $\wt{\hilbert}\otimes\C$, act unitarily on $\wt{\hilbert}\otimes\C$.
As the vector $\Psi_0$ from Lemma \ref{lemma:HegerfeldtRuijsenaars} we will take $\rho(0)$, which is a vector from the space $\Herm\(\hilbert\)$.
As a closed subspace $\hilbert'$ we take the subspace $\wt{\Pi}_{\lambda}\wt{\hilbert}\otimes\C$.

Therefore, we can apply Lemma \ref{lemma:HegerfeldtRuijsenaars} in this case too, obtaining that if the Hamiltonian is bounded from below, accurate measurements and preparations are impossible even for mixed states.
\end{proof}

\section{Possible objections}
\label{s:objections}

Despite the generality of Theorem~{\ref{thm:no-accuracy}} and especially of Theorem~{\ref{thm:no-accuracy-mixed}}, a mathematical result is as strong as its assumptions. Therefore, if one wants to reject these results to save both the Hamiltonian's boundedness from below and the accuracy of quantum measurements and preparations, one can try to reject these assumptions.

\begin{objection}[No perfect isolation]
\label{objection:environment}
Assumption~\ref{assumption:isolated} is not true, because the system composed of the measuring device and the observed system can never be perfectly isolated from the environment.
\end{objection}
\begin{proof}[Reply]
We can never be sure that a fluctuation from the environment doesn't happen, making the outcomes of the measurement accurate by a pure accident.
However, if we worry that the measurement is not perfectly isolated, we can take as our system $M$ the whole laboratory. We can even assume that the laboratory is on a spaceship in an inertial motion, far away from any other system. Or we can take the whole universe as our system, if we can assume that the universe is an isolated system.
Then, we replace the Hilbert space $\hilbert_{M}$ with $\hilbert_{M}\otimes\hilbert_{E}$, and the pointer observable $\obs{Z}$ with $\obs{Z}\otimes\obs{I}_{E}$, where $\hilbert_{E}$ represents the environment, and $\obs{I}_{E}$ is the identity operator on $\hilbert_{M}$.
Then, the proofs of Theorems~\ref{thm:no-accuracy} and ~\ref{thm:no-accuracy-mixed} remain unchanged.
\end{proof}

\begin{objection}[Time-dependence of the Hamiltonian]
\label{objection:time-dependent-Hamiltonian}
Lemma~\ref{lemma:HegerfeldtRuijsenaars} from Ref. \cite{HegerfeldtRuijsenaars1980RemarksOnCausalityLocalizationAndSpreadingOfWavePackets} assumes a time-independent Hamiltonian; however, the Hamiltonian is, in fact, time dependent due to the presence of the interaction term $\obs{H}_{S+M}$ in $\obs{H}=\obs{H}_{S}\otimes\obs{I}_{M}+\obs{I}_{S}\otimes\obs{H}_{M}+\mu\obs{H}_{S+M}$, which turns on at time 0 and turns off at time $t = T$. To treat the problem properly, this time dependence must be taken into account, as it is the central effect under investigation. If one does so, it will be found that the conclusions of Lemma~\ref{lemma:HegerfeldtRuijsenaars} do not hold.
\end{objection}
\begin{proof}[Reply]
Consider, for example, the spin measurement of a Silver atom using the Stern-Gerlach device, the interaction is considered irrelevant before the Silver atom enters the magnetic field of the device and after it leaves it. This is why many discussions of the measurement process use a Hamiltonian of the form $\obs{H}=\obs{H}_{S}\otimes\obs{I}_{M}+\obs{I}_{S}\otimes\obs{H}_{M}+\mu\obs{H}_{S+M}$, with $\mu=0$ outside of the time interval $[0,T]$. In practice, however, there is no perfect way to ensure that no interactions occur outside of {$[0,T]$}.

However, our proofs don't need to assume that the interaction is exactly zero before the measurement starts.
The Silver atom could very well already feel the magnetic field of the Stern-Gerlach device at $t=0$ or before, and yet, the total Hamiltonian can be time-independent, ensuring the validity of equation~\eqref{eq:premeasurement}.
We don't even need to assume that it is exactly zero for $t<0$ and $t>T$, because all we used from Assumption~\ref{assumption:stable} is that the pointer state is an eigenstate of the pointer observable, at least for a finite time interval $[T,T']$, as in equation~\eqref{eq:pointer:stable} or, more generally, as in equation~\eqref{eq:pointer:stable-mixed}.
The Hamiltonian can be assumed to be time-independent due to Assumption~\ref{assumption:isolated}, as explained in the Reply to Objection~\ref{objection:environment}.

Possible external influences that can make the Hamiltonian time-dependent are either noise from the environment, or controlled interactions that aim to restore the accuracy. In the first case, it is unlikely that external noise will restore the accuracy of the measurement, it will rather introduce more error. In the second case, intentional interactions can restore the accuracy of the measurement only if controlled by someone who already accurately knows the state for all the systems involved.
\end{proof}

\begin{objection}[Challenging pointer persistence]
\label{objection:pointer-persistence-unproven}
Assumption \ref{assumption:stable} is, well, an unproven assumption.  Such constructs have no place in a formal proof.
\end{objection}
\begin{proof}[Reply]
There are two points to unpack in this objection.
First, even if one thinks that Assumption~\ref{assumption:stable} is unproven, it still has a place in a formal proof. Every formal proof should transparently state the assumptions under which the result is derived, so that the domain of applicability of the result can be known.
Then, if one of the assumptions is rejected, the proof does not necessarily apply in that case.
What has no place in formal proof is not having assumptions, but having hidden assumptions.

So if one thinks that Assumption~{\ref{assumption:stable}} is wrong, and that we can do measurements even if their results are not temporarily recorded, one is free to reject the conclusion of Theorem \ref{thm:no-accuracy}. But one can't reject the proof itself on this basis alone, only the conclusion.

However, rejecting Assumption~{\ref{assumption:stable}} will not save the accuracy, quite the contrary:

\begin{quote}
\textbf{A perfectly accurate result that is not recorded can't be seen by the experimentalist or anyone else, so it is totally useless.}
\end{quote}

If Assumption~{\ref{assumption:stable}} is not made, the experimentalist can't trust the result, since it may be perturbed before being read. This would defeat the whole purpose of Objection~{\ref{objection:pointer-persistence-unproven}}, by making the result inaccurate.
\end{proof}

\begin{objection}
\label{objection:projection-helps}
The proof of Theorem~\ref{thm:no-accuracy} is based on the assumption of unitarity. But at the end the Projection Postulate is applied to reduce the superposition from equation~\eqref{eq:premeasurement} to a single term. So unitarity is violated at $t=T$, and Lemma~\ref{lemma:HegerfeldtRuijsenaars} can't be applied.
\end{objection}
\begin{proof}[Reply]
The proof derives a contradiction from the cases when no projection is needed, because the observed system was in an eigenstate of the observable.
See the details of the proof, Remark~\ref{rem:thm-explained}, and Section~\sref{s:methods}.
\end{proof}

\begin{question}
\label{question:no-collapse}
The proof of Theorem~\ref{thm:no-accuracy} is based on the condition~\eqref{eq:calibration-measurement}, respectively \eqref{eq:calibration-preparation}, and so it can be understood as a proof by \emph{reductio-ad-absurdum} that these conditions can't be satisfied independently of the initial state of the observed system.
And, indeed, there are situations when a separable state evolves into a separable state~\cite{Chowdhury2024EntanglementSeparabilityAndCorrelationTopologyOfInteractingQubits}.
\end{question}
\begin{proof}[Answer]
Indeed, one can take the proof of Theorem~\ref{thm:no-accuracy} as a proof by \emph{reductio-ad-absurdum} that conditions~\eqref{eq:calibration-measurement} and \eqref{eq:calibration-preparation} can't be satisfied independently of the initial state of the observed system.
That is, it is not excluded that the pointer state is not in the ``ready'' state $\varphi_{\varnothing}(0)$ for all possible initial states, but if the observed system is already in an eigenstate $\psi_{\lambda}$ of the observable, the pointer had to be in an eigenstate for $\lambda$, and not of $\varphi_{\varnothing}(0)$ as normally assumed. This will allow accuracy even if the Hamiltonian is bounded from below.
It will also have the advantage of ensuring a single outcome without breaking the unitary evolution with a projection or collapse.
Such a single-world unitary (no-collapse) proposal was shown in \cite{Stoica2012QMQuantumMeasurementAndInitialConditions} to require apparently very fine-tuned initial conditions that depend on the future measurement. In this case, both the observed system and the pointer state have to be, at $t=0$, eigenstates corresponding to the outcome $\lambda$. This appears as ``superdeterministic'' or ``retrocausal'', although it can also be interpreted as the result of a global consistency condition in the four-dimensional relativistic block universe \cite{Stoica2016OnTheWavefunctionCollapse,Stoica2021PostDeterminedBlockUniverse}. In \cite{Stoica2017TheUniverseRemembersNoWavefunctionCollapse,Stoica2021PostDeterminedBlockUniverse} it was shown that this is the only way to ensure the conservation laws (Everett's interpretation also ensures them, but for the totality of the branches, while in each branch conservation laws are violated).
In \cite{Stoica2024ObservationAsPhysication}, a new version of this single-world unitary proposal, which separates the observables for their physical meaning to eliminate the apparent superdeterminism or retrocausality, was proposed.
While, in the current context, it seems that rejecting conditions~\eqref{eq:calibration-measurement} and \eqref{eq:calibration-preparation} can allow for accuracy even if the Hamiltonian is bounded from below, it does not present a clear mechanism how this would be possible. By contrast, since unitarity without projection is a very strong constraint, it may restrict the accuracy even more. But the possibility from Question~\ref{question:no-collapse} is indeed open.
\end{proof}

\section{Discussion}
\label{s:discussion}

We have seen that if the Hamiltonian is bounded from below, accurate measurements and accurate preparations are impossible.
Therefore, if this is true, accurate reading and writing data, being measurements and preparations, are also impossible.
This doesn't prevent us from getting asymptotically closer to perfect accuracy, just from ever reaching it. This limitation is similar to the other limitations of quantum measurements reviewed in the Introduction, but if the Hamiltonian is bounded from below it is stronger, in some sense, than the Wigner-Araki-Yanase theorem \cite{wigner1952MessungQMOperatoren,Wigner1952MessungQMOperatorenPBusch2010EnTranslation,ArakiYanase1960MeasurementofQMOperators}, because it prevents measurement accuracy and non-disturbance independently and for all situations, and it does not need to use an additive conserved observable that does not commute with the observable $\obs{A}$.

Accuracy is important, especially if we want our technologies to be scalable and not to remain interesting experiments confined to laboratories or limited-purpose quantum computers. The resource costs to achieve even limited accuracy matter, because these resources are limited. This result affects quantum control \cite{DAlessandro2021IntroductionToQuantumControlAndDynamics}, quantum computing \cite{NielsenChuang2010QuantumComputationAndQuantumInformation}, and any quantum technology dependent on the accuracy of quantum preparations and measurements, inviting us to give more careful consideration to the limits imposed by the physical laws to the technological possibilities.

Busch and Jaeger advocated for the necessity to rethink the world based on unsharp observables, to embrace the ``quantum fuzziness'' of the world \cite{BuschJaeger2010UnsharpQuantumReality}.
Now we see that this becomes even more stringent.

Given that all that we know about the world, even its macroscopic properties, we know through observations, which are ultimately quantum, this would make everything we know about the external world to be inaccurate.
All we do in the world, even macroscopically, ultimately rely on quantum preparations, making our actions imprecise.
But is this really so? We certainly have very complex electronic devices, including programmable classical computers, and they seem to work very well. They also seem scalable, at least so far we didn't encounter any limitation of the memory they can have. While they are classical computers, they are also quantum systems, and not only because transistors employ quantum effects, but because the world is quantum. So we can consider the widespread and ubiquitous usage of information technology in our everyday life as the largest crowdsourced experiment that tested the accuracy of reading and writing bits. In fact, this experiment is a natural continuation of the longest crowdsourced experiment that showed that we can keep records of past events, for example by writing them down on a piece of paper. Even a mark on a tally stick is a physical state that can be distinguished from other physical states by observations that are, ultimately, quantum.
If the existence of recording devices or classical computers contradicts our results, could it be because the Hamiltonian is actually unbounded? We don't have a quantitative prediction of the possible accuracy in terms of the size or other properties of the measurement device, so it would be rushed to draw some conclusions from this. But we definitely need to explore this and find out.

At the same time, we may want to reconsider the assumption that our world's Hamiltonian is bounded from below, as an alternative way to regain some of the lost certainty.
This assumption leads to several results that may be problematic.
First, it was shown to forbid the existence of a time operator \cite{Pauli1980GeneralPrinciplesOfQuantumMechanics}, \cite{UnruhWald1989TimeAndTheInterpretationOfCanonicalQuantumGravity}.
Second, it was shown to imply that relativistic wave packets confined to a bounded region of space spread instantaneously in the entire space \cite{HegerfeldtRuijsenaars1980RemarksOnCausalityLocalizationAndSpreadingOfWavePackets}, \cite{Malament1996InDefenseOfDogmaWhyCannotBeARelativisticQuantumMechanicsOfParticles}.
Third, it implies that there is always a nonvanishing amplitude that the universe returns to a previous state \cite{Stoica2022ProblemOfIrreversibleChangeInQuantumMechanics}.
Finally, in this article it was shown that if the Hamiltonian is bounded from below, it prevents accurate measurements and preparations. 

Only time will tell how to solve this dilemma.
To avoid closing on a pessimistic note regarding accurate measurements and preparations, even if it turns out that they limit the possibility of scalable quantum computers, we have to remember what we learned from the history of quantum mechanics. Other results, initially perceived as limitations, are behind the incredible development of quantum technologies and other applications: results like complementarity, Bell's theorem \cite{Bell64BellTheorem}, the Kochen-Specker theorem \cite{KochenSpecker1967HiddenVariables}, the no cloning theorem \cite{Park1970NoCloning,WoottersZurek1982NoCloning} and so on, in the hands of visionary thinkers turned out to provide the revolutionary power of quantum technologies. Maybe these results rejected some of our initial expectations and intentions, but at the same time they opened unexpected avenues.

\textbf{Acknowledgments.} The author thanks Mohammad Hamed Mohammady and anonymous referees for valuable comments and suggestions offered to a previous version of the manuscript. 
Nevertheless, the author bears full responsibility for the article.

\section*{References}


\end{document}